\documentclass[twoside,leqno,twocolumn]{article}
\usepackage{ltexpprt}
\usepackage{sodafixes}

\usepackage{amsmath,amssymb}
\usepackage{graphicx}
\usepackage{cite}
\usepackage{hyperref}
\usepackage{microtype}

\newtheorem{observation}{Observation}[section]

\newtheorem{definitionhack}{Definition}[section]
\newenvironment{definition}{\begin{definitionhack}\normalfont}{\end{definitionhack}}

\title{Treetopes and their Graphs}
\author{David Eppstein\thanks{Computer Science Department, University of California, Irvine; Irvine, CA, USA. This material is based upon work supported by the National Science Foundation under Grant CCF-1228639 and by the Office of Naval Research under Grant No. N00014-08-1-1015.}}

\date{ }

\begin{document}
\maketitle

\begin{abstract}
We define treetopes, a generalization of the three-dimensional roofless polyhedra (Halin graphs) to arbitrary dimensions. Like roofless polyhedra, treetopes have a designated base facet such that every face of dimension greater than one intersects the base in more than one point. We prove an equivalent characterization of the $4$-treetopes using the concept of clustered planarity from graph drawing, and we use this characterization to recognize the graphs of $4$-treetopes in polynomial time. This result provides one of the first classes of 4-polytopes, other than pyramids and stacked polytopes, that can be recognized efficiently from their graphs.
\end{abstract}

\section{Introduction}
According to Steinitz's theorem~\cite{Ste-EMW-22}, the graphs of three-dimensional convex polyhedra can be characterized in purely graph-theoretic terms: they are the $3$-vertex-connected planar graphs (with more than three vertices). Based on this result, and a long line of research on algorithmic planarity testing \cite{HopTar-JACM-74,BooLue-JCSS-76,ChiNisAbe-JCSS-85,ShiHsu-TCS-99,BoyMyr-JGAA-04,FraOssRos-IJFCS-06}, it is possible to solve the recognition problem for graphs of $3$-polyhedra in linear time.
However, the situation for higher-dimensional polytopes is quite different.
Recognizing the face lattice of a polytope is complete for the existential theory of the reals, even for the special case of four-dimensional polytopes~\cite{RicZie-BAMS-95}. This puts the problem in a complexity class that, although solvable in polynomial space, is at least as hard as the NP-complete problems~\cite{Sch-GD-09}, and strongly suggests that recognition of the graphs of polytopes is also hard.

A natural response to this hardness result is to search for special classes of polytopes whose recognition problem is easier. One easy case is given by the four-dimensional pyramids (prisms over three-dimensional polyhedra): their graphs are \emph{apex graphs} (the graphs that can be made planar by deleting one vertex), and are easy to recognize by searching for a universal vertex and testing planarity and $3$-connectivity of the remaining graph. (Apex graphs in which the apex is not necessarily universal may also be recognized efficiently~\cite{Kaw-FOCS-09}.) The $d$-dimensional \emph{stacked polytopes}, formed by gluing simplices together on shared faces, have as their graphs the $(d+1)$-trees with the property that each $d$-clique is a subgraph of at most two $(d+1)$-cliques~\cite{KocPer-SW-76}; this characterization allows these polytopes to be recognized in polynomial time regardless of dimension. It is also possible to recognize the graphs of a class of generalized prisms, formed as the Cartesian product of any number of line segments, polygons, and three-dimensional polyhedra, in polynomial time~\cite{FeiHerSch-DAM-85,ImrPet-DM-07}. However, beyond these special cases and their combinations, very little is known.

In this paper we introduce another class of four-dimensional polytopes whose recognition problem is polynomially solvable, generalizing the case of the pyramids discussed above. The polytopes that we study are defined in terms of a designated \emph{base} facet, with the property that every face of dimension greater than one intersects the base in more than one point. We call these polytopes \emph{treetopes}, because the edges of the polytope that do not lie within the base must form a tree. The definition can be applied to polytopes of any dimension; the $3$-treetopes are exactly the polyhedral realizations of Halin graphs, and the treetopes of higher dimensions include all pyramids and  Cartesian products of pyramids with other polytopes.

We provide a graph-theoretic characterization of $4$-treetopes by relating them to a standard problem in graph drawing, \emph{clustered planarity}. In this drawing style, a planar graph is given together with a hierarchical clustering on its vertices. It must be drawn without crossings, representing the clusters as Jordan curves that surround only the parts of the graph in their cluster~\cite{FenCohEad-ESA-95,Dah-LATIN-98,GutJunLei-GD-02,CorDiB-SCG-05,CorDiBFra-JGAA-08}. We define a restricted type of clustering of planar graphs, which we call a \emph{well-connected clustering}, and we show that the graphs of $4$-treetopes are exactly the graphs that can be formed from a well-connected clustering by adding an additional vertex for each cluster. Based on this result, we also characterize the $4$-treetope graphs in terms of certain contraction and expansion operations (replacing a cluster by a single vertex or the inverse operation), and we use these operations to build a realization of any given $4$-treetope. Finally, we describe a polynomial-time recognition algorithm for the graphs of $4$-treetopes, which uses only the graph structure and not its geometry to find a valid sequence of the same contraction operations.

Our primary motivation for starting this research was to investigate the question: what is the higher-dimensional generalization of a Halin graph? We believe that treetopes are the answer to this question. However, the connection to clustered planarity provides us with a second motivation that comes from applications in information visualization. It is still unknown, and a major unsolved problem in graph drawing, whether clustered planar graphs can be recognized in polynomial time. Although the clusterings derived from $4$-treetopes are not difficult instances for the clustered planarity problem, our research on the graph theoretic properties of $4$-treetopes may lead to new insights that help solve this problem.

The rest of this paper is organized as follows.
In \autoref{sec:treetopes} we define treetopes, and prove some structural properties of their face lattices and their graphs that can be stated independently of their dimension.
In \autoref{sec:clustered-planarity} we turn to clustered planar graphs.
We define the cluster graph of a hierarchically clustered graph (a graph augmented by adding a new vertex for each cluster). We also define well-connected clusterings (clusterings that obey certain graph-theoretic properties analogous to the properties of treetopes), and we prove that these clusterings may be obtained by a sequence of operations in which we replace a single vertex of a graph by a new cluster.
In \autoref{sec:realization} we show how to realize each such expansion operation geometrically, proving that the cluster graphs of well-connected clusterings are exactly the graphs of 4-treetopes.
In \autoref{sec:recognition} we use the clustering-based characterization of these graphs to develop an algorithm for recognizing these graphs in polynomial time.
Finally, in \autoref{sec:properties} we discuss the sparsity and minor-containment properties of the graphs of 4-treetopes and of certain related clustered planar drawings.

\section{Treetopes}
\label{sec:treetopes}

\subsection{Definitions}

The following definitions are standard.
\begin{definition}
A \emph{polytope} is the convex hull of a finite set $S$ of points in a Euclidean space.
The \emph{faces} of a polytope are its intersections with halfspaces whose boundaries are disjoint from the relative interior of the polytope. They form a lattice in which the bottom element is the empty set and the top element is the polytope itself. The \emph{dimension} $\dim F$ of a face $F$ is one less than the minimum number of points of $S$ whose affine hull contains the face. In particular, by this definition, the dimension of the empty set is $-1$, and the faces of dimension zero (\emph{vertices}) are individual points that form a subset of $S$. The \emph{edges} of a polytope are its faces of dimension one. These are line segments, and the edges and vertices together form an undirected graph, the \emph{graph} or \emph{$1$-skeleton} of the polytope. The \emph{facets} are the faces of dimension one less than the polytope.
\end{definition}

% Place figure early for proper placement in two-column layout
\begin{figure}[t]
\centering\includegraphics[width=0.45\textwidth]{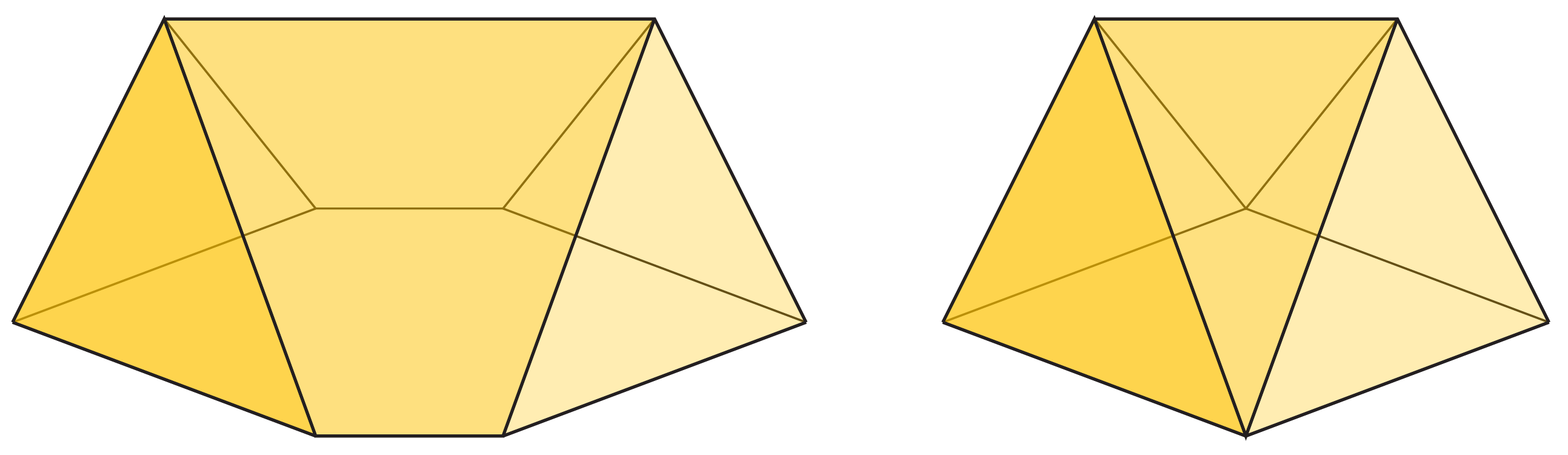}
\caption{Left: a $3$-treetope. Right: a $3$-polytope that is not a treetope (some $2$-faces intersect the base in only one vertex) but in which the faces disjoint from the base all have dimension at most one.}
\label{fig:non-treetope}
\end{figure}

The next definition is our main object of study.
\begin{definition}
We define a \emph{treetope} to be a polytope in which there exists a distinguished facet $B$ (the
\emph{base}) with the property that every face that intersects $B$ in at most one point has dimension at most one. A $k$-treetope is a treetope of dimension $k$.
\end{definition}
In particular we will be interested in $4$-treetopes.

The $3$-treetopes have long been studied~\cite{Kir-PTRSL-56} and have been called \emph{based polyhedra}~\cite{Rad-IJM-65}, \emph{roofless polyhedra}~\cite{CorNadPul-MP-83}, or \emph{domes}~\cite{DemDemUeh-CCCG-13}. Their graphs are the \emph{Halin graphs}, the graphs formed from a planar embedding of a tree without degree-two vertices by adding a cycle that connects the leaves of the tree in the order given by the embedding~\cite{Hal-CMA-71,CorNadPul-MP-83}.
Recently, we made a more general study of the polytopes in which the faces disjoint from a base facet all have bounded dimension~\cite{EppLof-DCG-13}.
As we observed, when the dimension bound is one, the faces that are disjoint from the base form a tree. However, there exist polytopes that meet this definition but are not treetopes (\autoref{fig:non-treetope}).

Examples of treetopes with arbitrary dimension $d$ are given by the \emph{pyramids} over $(d-1)$-dimensional bases, polytopes formed by the convex hull of the base plus one more vertex (the \emph{apex}) that is in general position with respect to the base. \autoref{fig:pyr-prism}, left, gives an example in which the base is a three-dimensional cube. In the figure, this cubical pyramid is projected into a three-dimensional \emph{Schlegel diagram} in which apex is shown as the point in the center of the cube. In a pyramid, every face of dimension greater than one contains two or more base vertices, because there is only one non-base vertex to include, so these shapes necessarily meet the definition of a treetope.

\begin{figure}[t]
\centering\includegraphics[scale=0.35]{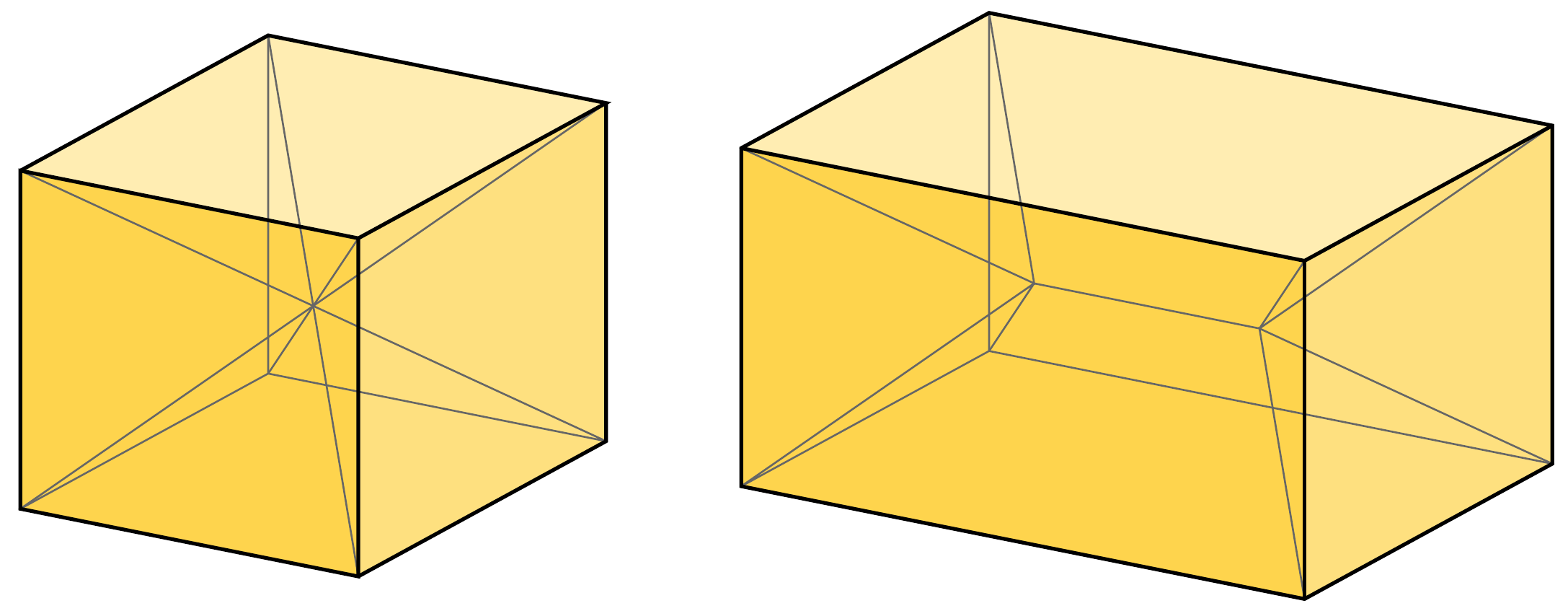}
\caption{Schlegel diagrams of the cubical pyramid (left) and the square-pyramidal prism (right). In both cases the outer face of the diagram can be taken as the base of a treetope.}
\label{fig:pyr-prism}
\end{figure}

 Another class of examples of treetopes is given by the \emph{prisms} over $(d-1)$-dimensional pyramids; that is, the polytopes formed as the Cartesian product $P=Q\times I$ of a pyramid $Q$ with a line segment~$[0,1]$ (\autoref{fig:pyr-prism}, right). If $B$ is the base of $Q$, then $B\times I$ may be taken as the base of $P$; a simple case analysis shows that with this choice of base $P$ is a treetope. The same analysis shows that the Cartesian products of pyramids with any other polytopes are again treetopes. However, it is not necessarily true that prisms over treetopes remain treetopes. For instance, the two-dimensional square is a treetope (as is every two-dimensional polygon) but the prism over the square (the three-dimensional cube) is not a treetope.

\begin{definition}
We say that a treetope is in \emph{general position} if no two vertices have equal nonzero distance from its base.
\end{definition}
For the purposes of understanding the combinatorial structure of treetopes we may assume without loss of generality that any treetope is in general position, for if not it can be perturbed into general position by a projective transformation that does not change its combinatorial structure.

\subsection{Face structure}

In this section we examine the face structure of treetopes. As we will show, the edges that do not lie in the base form a tree, justifying the name. This tree also has a close connection to the rest of the faces.

\begin{definition}
If $v$ is any vertex of a treetope $P$, we define a \emph{parent} of $v$ to be a vertex adjacent to $v$ in the graph of $P$ and farther than $v$ from the base hyperplane of $P$, and we define a \emph{root} of $P$ to be a vertex with no parent. (Shown for a $3$-treetope in \autoref{fig:Topographic}.)
\end{definition}

\begin{definition}
Let $v$ be a vertex of a $d$-dimensional polytope. Then the \emph{link} of $v$ is the $(d-1)$-dimensional polytope formed by intersecting $P$ with any hyperplane that is disjoint from $v$ but passes through the interior of $P$ near enough to $v$ so that all other vertices of $P$ are on the other side of the hyperplane. The precise geometry of the link depends on the choice of hyperplane, but its combinatorial structure does not.
\end{definition}

\begin{figure}[t]
\centering
\includegraphics[width=0.35\textwidth]{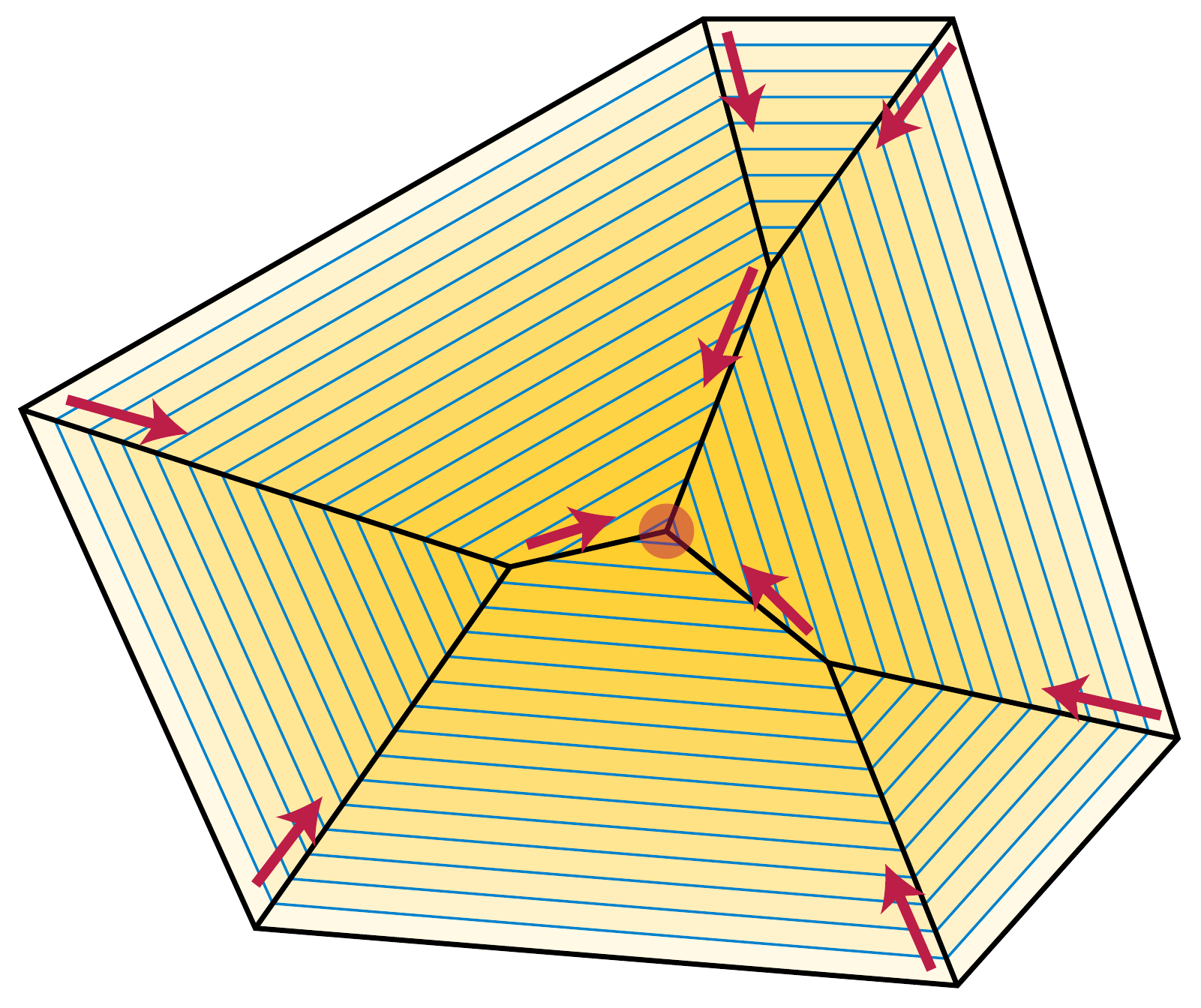}
\caption{Top view of a $3$-treetope, projected perpendicularly onto its base face, and showing level sets for the distance from the base plane. The root vertex (farthest from the base) is marked by a red disk; the red arrows show the unique parent (a vertex farther from the base) for each non-root vertex.}
\label{fig:Topographic}
\end{figure}

\begin{lemma}
\label{lem:parent}
For a treetope in general position, every vertex has at most one parent and there is exactly one root.\end{lemma}

\begin{proof}
Let $P$ be a given treetope with base $B$ and canopy~$T$.
The fact that there is only one root follows from the simplex method in linear programming, which can be used to find the maximum of any linear function (such as the function mapping each point in $P$ to its distance from the hyperplane of $B$) by following a path in the graph of $P$ along which the function is monotonically increasing. The only vertex that can be a root is the maximum of~$f$ (unique by the assumption of general position), because for any other vertex the simplex method will find a parent edge as the first edge of its path.

To see that each vertex $v$ in $T$ can have at most one parent, consider the link of $v$. Each vertex of the link corresponds to an edge incident to $v$, and for each such edge all points on the edge lie above $v$ or below $v$ with respect to $f$. Therefore, each vertex of the link may be seen as being above $v$ or below $v$, without regard to the specific hyperplane near $v$ that was used to form the link. Similarly, there are three cases for each face of the link: a face may be entirely above $v$, it may be entirely below $v$, or it may contain vertices of the link that are both above and below $v$. By the same simplex-algorithm argument the faces of the link that are above $v$ with respect to $f$ form a connected complex. If $v$ could have more than one increasing edge, this complex would have more than one vertex, and hence would have at least one edge. This edge in the link would necessarily correspond to a two-dimensional face $F$ of $P$ within which $v$ is the unique minimum point of~$f$. But then $F\cap B$ would either equal $v$ (if $v$ belongs to $B$) or be empt (otherwise), contradicting the assumption that no faces of dimension two or more intersect $B$ in at most one point.
\end{proof}

\begin{corollary}
\label{cor:canopy-is-tree}
Let $P$ be a treetope of dimension $d$ with base $B$, and let $T$ be the set of faces of $P$ that intersect $B$ in at most one point. Then $T$ is an unrooted tree, each leaf of $T$ lies in $B$, and each non-leaf of $T$ has degree at least~$d$.
\end{corollary}

\begin{proof}
The fact that $T$ is a tree follows from \autoref{lem:parent}.
The claim about the degrees of the non-leaf vertices of~$T$ follows from the fact that the vertex figure of any vertex in a $d$-polytope is a $(d-1)$-polytope, which necessarily has at least $d$ vertices.
\end{proof}

\begin{definition}
For a treetope $P$ with base $B$ and tree $T$ defined as above, we call $T$ the \emph{canopy} of $P$.
\end{definition}

\begin{lemma}
\label{lem:facebase}
Let $P$ be a treetope with base $B$. Then for every face $F$ of $P$ of dimension greater than one, either $F\subset B$ or $F$ is a treetope with base $F\cap B$. In particular, $\dim(F\cap B)\ge \dim F - 1$.
\end{lemma}

\begin{proof}
We prove first the claim about the dimension. By the definition of treetopes, $F\cap B$ contains at least two vertices; let $v$ be one such vertex.
Then $F$ lies within the positive hull of the edges of $F$ incident to $v$. By \autoref{lem:parent} all but at most one of those edges lies within $F\cap B$, so the dimension of the positive hull (and therefore of $F$) is at most one more than the dimension of $F\cap B$.

Now suppose that $F$ is not a subset of $B$. Then, by the dimension claim, $F\cap B$ is a facet of~$F$.
Let $G$ be a face of $F$ such that $G\cap (F\cap B)$ is a single vertex. Then by associativity $(G\cap F)\cap B=G\cap B$ is the same single vertex, and by the assumption that $P$ is a treetope $G$ has dimension at most one. Thus, the faces of $F$ have the defining property of treetopes.
\end{proof}

For instance, every non-base facet of a $4$-treetope must be a roofless polyhedron.

\begin{lemma}
\label{lem:lift}
Let $P$ be a treetope with base $B$, and let $F$ be a nonempty face of $B$. Then there is exactly one face $F'$ of $P$ such that $F'\ne F$ and $F'\cap B=F$.
\end{lemma}

\begin{proof}
A face $F'$ meeting the description of the property may be found by starting with $F'=P$ and then, as long as the intersection of $F'$ with $B$ is of too large a dimension, replacing $F'$ by one of its facets. Each such step reduces the dimension of the intersection with $B$ by one unit, by \autoref{lem:facebase}, so it is not possible for this sequence of steps to skip over~$F$.

Let $F'$ be any such face, and let $v$ be any vertex of $F$. Then $F'$ must lie in the affine hull of $F$ and the parent edge of $v$. The dimension of this affine hull equals the dimension of $F'$, so it equals the affine hull of $F'$ itself. However, any two different faces of a polytope must have different affine hulls, so $F'$ is the unique face satisfying the properties required by the lemma.
\end{proof}

\begin{definition}
For the faces $F$ and $F'$ of \autoref{lem:lift}, we say that $F$ is the \emph{base} of~$F'$ and that $F'$ is the \emph{lift} of~$F$.
\end{definition}

\begin{lemma}
\label{lem:lift-ancestor}
Let $F'$ be the lift of a face $F$ in a treetope~$P$ with canopy $T$, choose a root for $T$ at a leaf vertex that does not belong to~$F$, and let $a$ be the lowest common ancestor in $T$ of the vertices of~$F$. Then the canopy of $F'$ is the union of the paths in~$T$ between $a$ and the vertices of~$F$.\end{lemma}

\begin{proof}
Every edge of the canopy of $F'$ belongs to a $2$-face of $F'$, so the canopy of~$F'$ equals the union of the canopies of the 2-faces of~$F'$. By \autoref{lem:lift} every $2$-face is the lift of an edge $uv$ of $F$. The canopy of this $2$-face consists of the unique path in $T$ from $u$ to $v$. This path belongs to the union of paths described in the lemma, so the canopy of $F'$ is a subset of the union of paths.

To show that it equals the union of paths, let $v$ be any vertex of $F$ and $e$ be any edge on the path from $v$ to $a$. Thus, $e$ is an arbitrary edge in the union of paths, and we must show that $e$ belongs to the canopy of $F'$. Let $w$ be an arbitrary descendant of $a$ in $F$ through a different child than the one leading to~$v$. Then the path in $T$ from $v$ to $w$ passes through $e$. However, $vw$ might not be an edge of $F$ and this path might not be the boundary face of a $2$-face in $F'$.
Nevertheless, because $F$ is connected, there exists a path $\pi$ in $F$ from $v$ to $w$. Choose such a path arbitrarily and let $vu$ be the first edge on this path. If $u$ is connected to $a$ through a path that does not contain $e$, then the path from $u$ to $v$ does contain~$e$, and we have found a $2$-face (the lift of~$uv$) that contains~$e$. If the path from $u$ to $a$ does contain~$e$, then the pair of vertices $u$ and $w$ are connected in~$T$ by a path containing $e$, and are connected in~$F$ by a shorter path than $\pi$. In this case the result follows by induction on the length of~$\pi$.
\end{proof}

We may summarize the results in this subsection as a theorem:

\begin{theorem}
\label{thm:classify-faces}
Let $P$ be a treetope with base $B$. Then the edges of $P$ that do not lie in $P$ form a tree, the \emph{canopy} $T$ of $P$.
The faces of $P$ may be partitioned into three classes:
\begin{enumerate}
\item the faces of $B$,
\item the edges of $T$ that are disjoint from~$B$, and
\item one face $F'$ of dimension $i+1$ for each $i$-dimensional face $F$ of $B$, called the \emph{lift} of $F$.
\end{enumerate}
For each face $F$ of $B$ with lift $F'$, $F'$ is a treetope with base $F$, and the canopy of $F'$ is
the minimal subtree of $T$ that connects all the vertices in $F$.
\end{theorem}

\begin{corollary}
Every $4$-treetope with $n$ vertices or with $n$ facets has $O(n)$ faces.
\end{corollary}

\begin{proof}
In any $4$-treetope, the facets of the base correspond one-to-one with the non-base facets of the treetope, so if the treetope has $n$ facets then the base has $n-1$ facets.
The base has $O(n)$ faces, because it is a $3$-polytope and every $3$-polytope has a number of faces that is linear in either its number of vertices or its number of facets. In particular, if the treetope has $n$ facets then the base has $O(n)$ vertices. The canopy has $O(n)$ vertices and edges, because it is a tree with $O(n)$ leaves and no degree-two internal vertices. And there are $O(n)$ remaining faces, because each is the lift of a unique face of the base and there are $O(n)$ base faces.
\end{proof}

\subsection{Branches and slices}

\begin{figure}[b]
\centering
\includegraphics[width=0.35\textwidth]{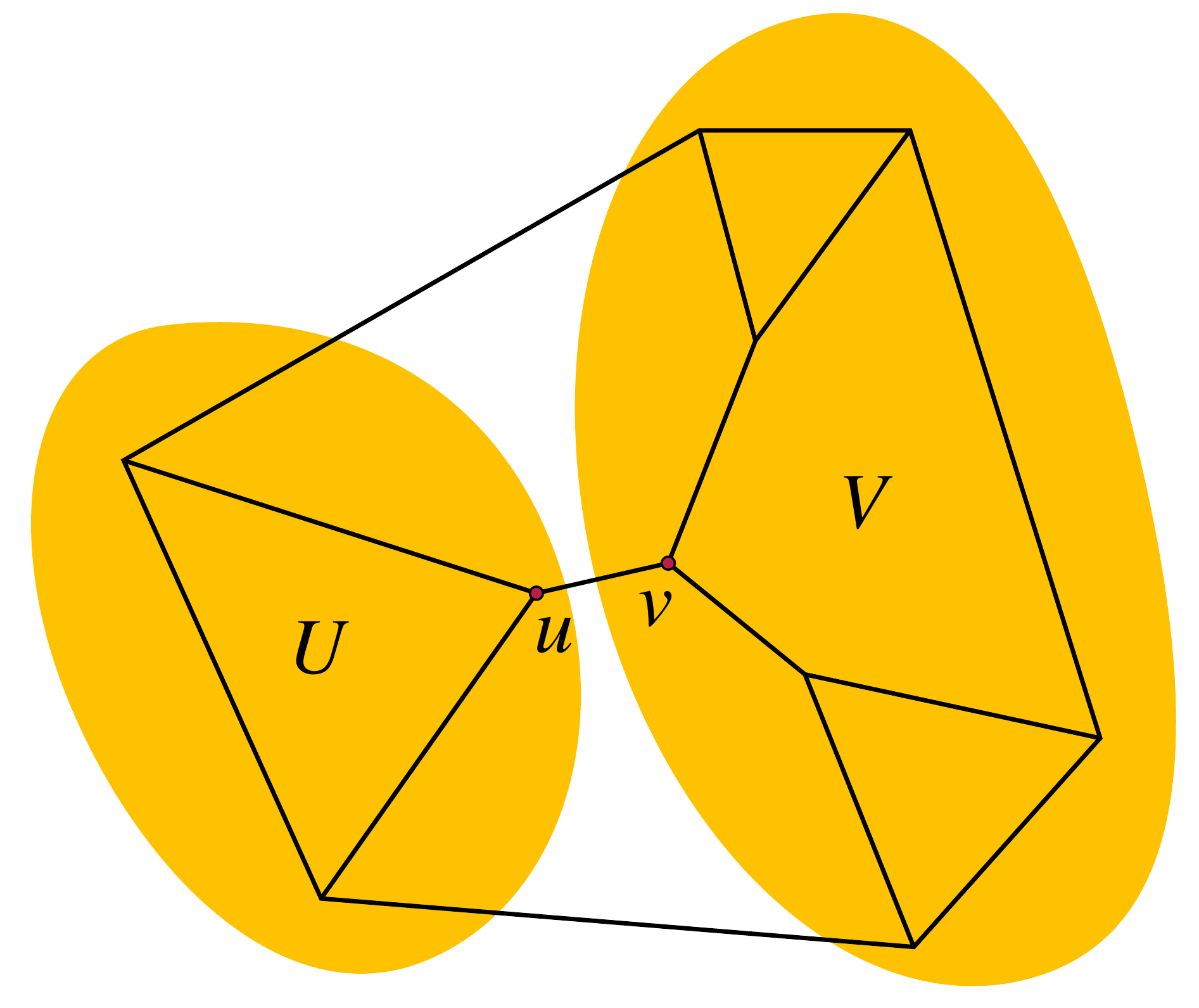}
\caption{A slice of the $3$-treetope of \autoref{fig:Topographic} determined by an edge $uv$, partitioning it into two complementary branches $U$ and $V$.}
\label{fig:Slice}
\end{figure}

\begin{definition}
Suppose that $P$ is a treetope with base~$B$, and $uv$ is an edge of $P$ such that neither $u$ nor $v$ belongs to $B$. Then we can partition the canopy into two subtrees by deleting edge $uv$.
Let $U$ be the subset of the vertices of $P$ in the subtree containing $u$, and $V$ be the subset of the vertices of $P$ in the other subtree containing $v$.
Then we call $U$ and $V$ \emph{branches} of $P$, and we call the partition $(U,V)$ of the vertices of~$P$ into two complementary branches a \emph{slice} of $P$. An example of a slice and its two branches is shown in \autoref{fig:Slice}.
\end{definition}

\begin{observation}
Each branch must have at least $\dim P$ vertices.
\end{observation}

\begin{proof}
The branch contains one vertex for the endpoint~$u$ of edge $uv$ defining the branch, and another vertex for each neighbor of $u$ other than $v$. A sufficient bound on the number of neighbors is given in \autoref{cor:canopy-is-tree}.
\end{proof}

\begin{definition}
Let $P$ be a polytope with base $B$, and $(U,V)$ be a slice of $P$ defined by canopy edge $uv$.
Then we define the \emph{stems} of branch $U$ to be subsets of $U$, one for each edge $uw$ where $w\ne v$. If $w\in B$, then its stem is the singleton set $\{w\}$. Otherwise, its stem is the branch~$W$ of the slice $(X,W)$ defined by edge $uw$.
\end{definition}

\begin{definition}
We say that a subset $S$ of the vertices of $B$ is \emph{externally $k$-connected}
if the graph formed from the graph of $B$ by contracting all vertices of $B\setminus S$ into a single supervertex is $k$-vertex-connected.
\end{definition}

\begin{lemma}
\label{lem:externally-connected}
Let $P$ be a treetope with base $B$, and $(U,V)$ be a slice of $P$ defined by canopy edge $uv$.
Then $U\cap B$ is externally $(\dim P-1)$-vertex-connected.
\end{lemma}

% Way out of place in source so that this ends up only mildly out of place in output
\begin{figure*}[t]
\centering\includegraphics[width=0.95\textwidth]{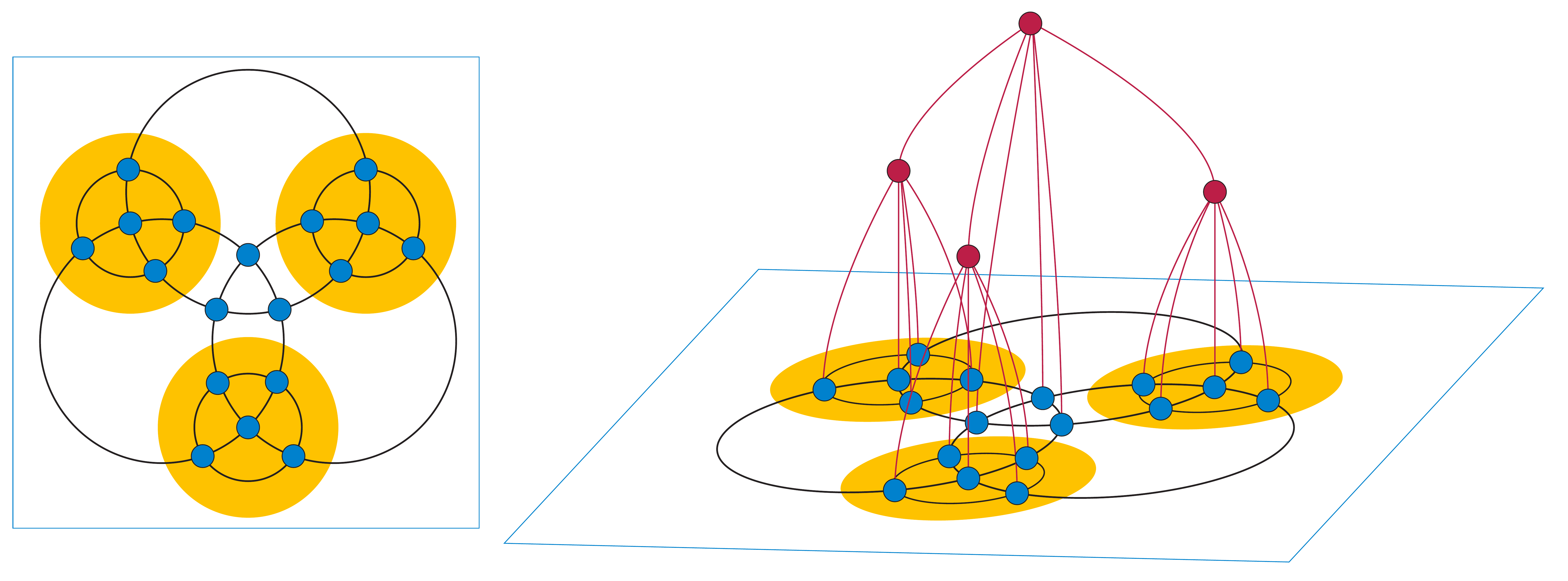}
\caption{A well-connected clustering (left) and its cluster graph (right).}
\label{fig:clustergraph}
\end{figure*}

\begin{proof}
Let $d=\dim P$, and define a graph $G$ by contracting all vertices of $P\setminus (U\cap B)$ into a single supervertex. Let $K$ be a set of at most $d-2$ vertices in~$G$. We prove by induction on the cardinality of~$U$ that deleting $K$ from $G$ leaves a connected graph. By the induction hypothesis, each stem of $U$ is externally $(d-1)$-connected, and each vertex in $K$ corresponds to at most one vertex in the contracted graph for each stem. It follows that the deletion of $K$ cannot disconnect the vertices within any stem. It remains to show that each two stems remain connected to each other; that is, that the graph formed from $G$ by contracting the remaining vertices of each stem into a single supervertex is connected.

Let $L$ be the link of $u$, a polytope formed from $P$ by intersecting it with a hyperplane near $u$. This is a $(d-1)$-dimensional polytope, whose faces are in one-to-one correspondence with the faces of $P$ that are incident to $u$. This correspondence changes the dimension of a face by one, so that an edge of the link corresponds to a $2$-face of $P$, etc. Observe that, if we were not deleting the vertices in $K$, then the graph of $L$ is isomorphic to the graph formed by contracting each stem, and the complementary branch~$V$, to a single vertex. For, each edge between two stems or between a stem and branch~$V$ lifts to a $2$-face of $P$ incident to $u$ (by \autoref{lem:lift-ancestor}) and therefore corresponds to an edge of $L$, and vice versa.

By Balinski's theorem~\cite{Bal-PJM-61}, the graph of $L$ is $(d-1)$-vertex-connected. Deleting a vertex in $K$ may change this graph either by damaging the complementary branch~$V$ (which we cannot assume to remain connected after the deletion because $V$ might not come before $U$ in the induction order) or by removing the endpoint of one of the edges linking two stems. However, as there are only $d-2$ deletions, the graph remains connected after this damage, and therefore no two stems can be separated from each other.
\end{proof}

\begin{corollary}
For any branch $U$, the subgraph of the graph of $B$ induced by $U\cap B$ is connected.
\end{corollary}

\begin{lemma}
\label{lem:cross-matching}
Let $(U,V)$ be a slice of $P$, and let $X$ and~$Y$ be two sets that are either stems of $U$ or the set $V$. Then there is at most one edge in $B$ connecting $X$ to $Y$.
\end{lemma}

\begin{proof}
As in the proof of \autoref{lem:externally-connected}, consider the link $L$ of $u$.
It has one vertex for each stem, and one vertex for~$V$. Every edge in $B$ connecting $X$ to $Y$ lifts to a $2$-face of $P$ that passes through $u$. This $2$-face in turn corresponds to an edge between the two vertices in $L$ that correspond to $X$ and $Y$. Distinct edges in $B$ lift to distinct edges in $L$, but each pair of vertices in $L$ can be the endpoints of at most one edge. Therefore there can be at most one edge from $X$ to $Y$ in $B$.
\end{proof}

For convenience, we again summarize the results of this section in a single theorem, describing the graph-theoretic properties of branches. We will use these properties to characterize the graphs of $4$-treetopes in an algorithmically recognizable way.

\begin{theorem}
\label{thm:slice}
If $P$ is a treetope with base $B$, and $(U,V)$ is a slice of $P$, then both $U$ and $V$ include at least $d-1$ stems. Each pair of stems of $U$ (or one stem and the set $V$) are connected by at most one edge in $B$. Both $U\cap B$ and $V\cap B$ are externally $(\dim P-1)$-vertex-connected.
\end{theorem}

\section{Clustered planarity}
\label{sec:clustered-planarity}

A long thread of research in the graph drawing community concerns \emph{clustered planarity}: given as input a pair $(C,G)$ where $G$ is a planar graph and $C$ is a nested family of subsets of the vertices of $G$ (a \emph{clustering} of~$G$), find a drawing of $G$ such that each cluster can be drawn as a simple closed curve surrounding its vertices, without crossings between pairs of clusters or between clusters and unrelated edges~\cite{FenCohEad-ESA-95,Dah-LATIN-98,GutJunLei-GD-02,CorDiB-SCG-05,CorDiBFra-JGAA-08}. It remains unknown whether clustered planarity can be tested in linear time, and so researchers have instead sought classes of instances that are general enough to cover the problems that might arise in practice but special enough that they can still be solved efficiently. We will take a different tack: we define a class of instances for the clustered planarity problem that are quite special, special enough to make the clustered planarity problem itself trivial for these instances. Instead we will construct non-planar graphs from these clustered planarity instances by adding a representative vertex for each cluster, and we will use this construction to characterize the graphs of $4$-treetopes.

\subsection{Definitions}

\begin{definition}
A \emph{polyhedral graph} is a $3$-vertex-connected planar graph with four or more vertices. By Steinitz's theorem a graph is polyhedral if and only if it is the graph of a $3$-polyhedron.
\end{definition}

\begin{definition}
If $G$ is a graph, and $C$ is a collection of subsets of the vertices of $G$, we say that $C$ is \emph{nested}, and that $(C,G)$ is a \emph{clustering} of $G$, if for every two sets $X$ and $Y$ in $C$ either $X\subset Y$, $Y\subset X$, or $X\cap Y=\emptyset$.
\end{definition}

Rather than representing clusterings as planar embeddings with the clusters drawn as simple closed curves we instead represent the clusters themselves as vertices in a larger graph. More precisely:

\begin{figure}[t]
\centering
\includegraphics[width=0.35\textwidth]{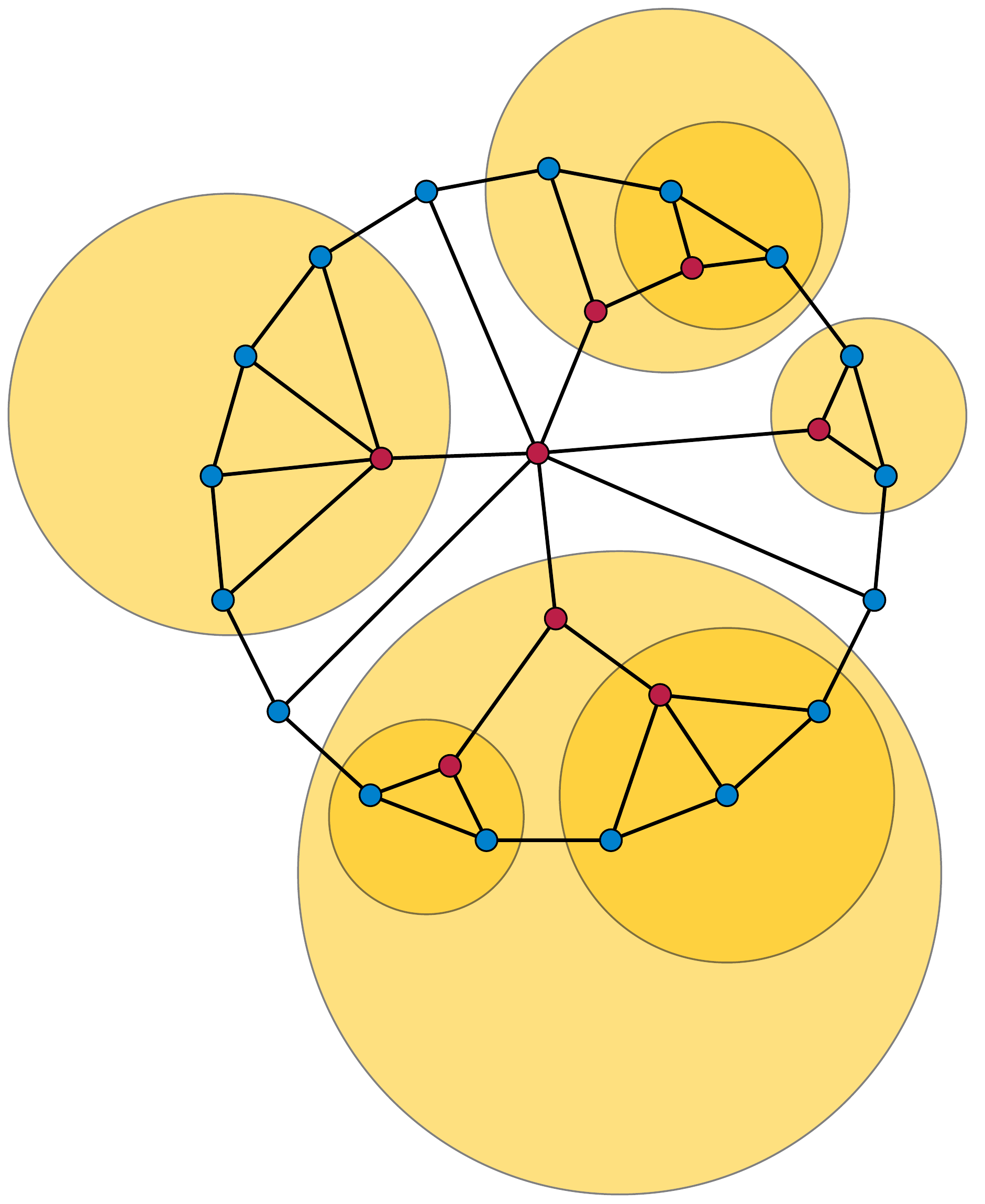}
\caption{For a cycle graph (blue vertices) with connected clusters (yellow disks), the cluster graph (with an added red vertex in each region formed by the circles) is a Halin graph, and any Halin graph can be formed in this way.}
\label{fig:ClusteredHalin}
\end{figure}

\begin{definition}
Let $(C,G)$ be a clustering, without any complementary pair of clusters, and (if it is not already in $C$) add $V(G)$ as a cluster in $C$.
Then we define the \emph{cluster graph} of $(C,G)$ to be a graph
that contains $G$ as a subgraph, and in addition has one vertex $c_X$ for each cluster $X$ in $C$.
Each vertex $v$ in $G$ is connected by an edge to the cluster vertex $c_X$ for the smallest cluster that contains~$v$. There always exists at least one such cluster because of the inclusion of $V(G)$ as a cluster. Each cluster vertex $c_X$ (other than the one for $V(G)$) is connected by an edge to the cluster vertex $c_Y$ for the smallest cluster that forms a strict superset of~$X$.
\end{definition}

The same construction may be represented topologically rather than combinatorially.
Let $(C,G)$ be any clustering. Represent the vertices of $G$ as points in the plane, and the nontrivial clusters of $G$ as Jordan curves disjoint from each other and the points, with each cluster consisting of the points inside the corresponding curve (ignoring whether the edges can be routed to give a valid clustered planar drawing). Then the cluster graph has a cluster vertex for each region into which the plane is divided by these curves, adjacent to the vertices for adjacent regions and to the points within its region. For example, if $G$ is a cycle graph and $C$ is a nested collection of  paths of two or more vertices in $G$, the resulting cluster graph is a Halin graph, and any Halin graph can be formed as a cluster graph in this way (\autoref{fig:ClusteredHalin}).

The following criteria for a more special class of clusterings and cluster graphs   (depicted in \autoref{fig:clustergraph}) are motivated by the properties described in \autoref{thm:slice}.

\begin{definition}
If $(C,G)$ is a clustering of a polyhedral graph, we say that $(C,G)$ is a \emph{well-connected clustering} if it has the following properties:
\begin{itemize}
\item Each cluster vertex $c_X$ in the cluster graph has degree at least four.
\item For each two disjoint sets $X$ and $Y$ that are either clusters, complements of clusters, or singleton vertex sets, and whose union is not the entire vertex set, at most one edge of $G$ has one endpoint in $X$ and one endpoint in $Y$.
\item For each cluster $X$ in $C$, other than the set of all vertices, and for the complementary set $Y=V(G)\setminus X$, both $X$ and $Y$ are externally $3$-vertex-connected in~$G$.
\end{itemize}
\end{definition}

For instance, the clustering shown in \autoref{fig:clustergraph} is well-connected. However, if the three central vertices were grouped into another cluster, the result would not be well-connected, because then there would exist disjoint but non-complementary pairs of clusters connected by more than one edge.
These definitions have been set up in such a way as to make the following observation clear:

\begin{observation}
\label{obs:treetope-cluster}
If $P$ is a $4$-treetope with base $B$, then the graph of $P$ is the cluster graph of a well-connected clustering $(C,G)$ where $G$ is the graph of~$B$.
\end{observation}

\begin{proof}
To form a well-connected clustering from $P$, we choose an arbitrary vertex $v\in B$ and define a cluster for each slice $(U,V)$, where the cluster is the intersection of $B$ with the branch of the slice that does not contain~$v$. The resulting clusters are  nested and their well-connectedness follows from \autoref{thm:slice}. Each vertex $u$ that is in one of the defined clusters is connected in the cluster graph to the vertex for the smallest cluster that contains it, which corresponds to the parent of $u$ in $P$. Each vertex that is not in any of these clusters (including $v$ itself) is connected to the cluster graph vertex corresponding to the cluster of all vertices in $G$, which again corresponds to its parent. Thus, the cluster graph and the graph of $P$ are isomorphic.
\end{proof}

\begin{definition}
The graph $G$ of every $3$-polyhedron~$P$ has a well-connected clustering with no nontrivial clusters, which represents the $4$-treetope formed as the pyramid over~$P$.
We call this the \emph{trivial clustering} of $G$.
\end{definition}

For some polyhedral graphs, the trivial clustering is the only well-connected clustering. For instance, this is the case for graphs such as the octahedral graph in which each $2$-face is a triangle, for in these graphs every partition of the vertices into two connected subsets (a cluster and its complement) has two edges that cross the partition and share an endpoint. We do not know whether the existence of a non-trivial well-connected clustering for a given polyhedral graph can be tested efficiently.

\subsection{Expansion and contraction}

\begin{definition}
Let $(C,G)$ be a well-connected clustering with at least one nontrivial cluster, and let $X$ be a cluster in $C$ that is not a superset of any smaller clustering. Then the \emph{contraction} of $X$ is the clustering $(C',G')$ in which we remove $X$ from the clustering and replace the vertices of $X$ in $G$ by a single supervertex, keeping all adjacencies to vertices outside $C$. The other clusters containing vertices of $X$ should also be modified in the obvious way, by replacing these vertices by the new supervertex.
\end{definition}

\begin{lemma}
With $C$, $G$, and $X$ as above, the contraction of $X$ is another well-connected clustering.
\end{lemma}

\begin{proof}
$G'$ remains polyhedral: it is $3$-vertex-connected and has at least four vertices by the external connectivity of $V\setminus X$.
The contraction does not change the required properties of any of the other clusters in $C'$.
\end{proof}

We define an \emph{expansion} to be the opposite operation to a contraction. More precisely:

\begin{definition}
Let $(C,G)$ be a well-connected clustering, let $v$ be a designated vertex in $G$, and let $H$ be a polyhedral graph containing a vertex $v'$ of the same degree as $v$. Additionally, suppose that we have identified a one-to-one order-preserving correspondence between the  edges incident to $v$
(in the cyclic order given by the embedding of~$G$) and the edges incident to $v'$ (in the cyclic order given by the embedding of $H$). Then we define the \emph{expansion} of $v$ by $H$ to
be a graph formed from $G$ by deleting $v$, adding $H-v'$ in its place, and reconnecting each of the edges that was incident to $v$ in $G$ to the corresponding neighbor of $v'$ in $H$.
We then add to $C$ another cluster for the vertices in $H-v'$ that were added to the graph, and modify the existing clusters in $C$ in the obvious way, by replacing $v$ in each cluster that contains it by the vertices of $H-v'$.
\end{definition}

For instance, the graph in \autoref{fig:clustergraph} can be formed by starting with the ($4$-regular $6$-vertex) graph of an octahedron and its trivial clustering, and then performing three expansions, each of which uses the graph of the octahedron as~$H$ and creates one of the three nontrivial clusters in the figure.

\begin{observation}
Expansions and contractions are inverse to each other: if we expand a vertex and then contract the resulting new cluster, or if we contract a cluster and then expand the resulting vertex by the graph defining the property of external connectivity of the contracted cluster, the result in either case is the original clustering.
\end{observation}

\begin{lemma}
With $C$, $G$, $v$, and $H$ as above, the expansion of $v$ by $H$ is another well-connected clustering.
\end{lemma}

\begin{proof}
The new cluster has the required degree in the cluster graph, because of the definitional requirement that the polyhedral graph $H$ has at least four vertices. It is externally $3$-connected by $3$-connectivity of both $G$ and $H$. Its addition does not change the cluster graph degree, or external $3$-connectivity of the other clusters, nor can it cause two edges to share endpoints when they did not do so previously. And the overall graph remains $3$-connected, because the change cannot introduce any new $2$-vertex cuts.
\end{proof}

We summarize the results of this section in a theorem:

\begin{theorem}
\label{thm:expansion}
The well-connected clusterings are exactly the clusterings that can be obtained from the trivial clustering of a polyhedral graph by a sequence of expansion operations.
Every expansion operation can be undone by a contraction operation, and vice versa.
Both expansion and contraction preserve the property of being a well-connected clustering.
\end{theorem}

\section{Realization}
\label{sec:realization}

In this section we prove that the cluster graphs of well-connected clusterings can always be realized as $4$-treetopes.

\subsection{Three-dimensional analogue}

It is dangerous to reason about higher-dimensional geometry by analogy to lower dimensions. Nevertheless, as an analogue of what we want to prove, consider the proof below of the following proposition, the special case of Steinitz's theorem for Halin graphs.

\begin{proposition}
\label{prop:realize-halin}
Every Halin graph can be realized as a $3$-treetope.
\end{proposition}

\begin{figure}[t]
\centering\includegraphics[scale=0.35]{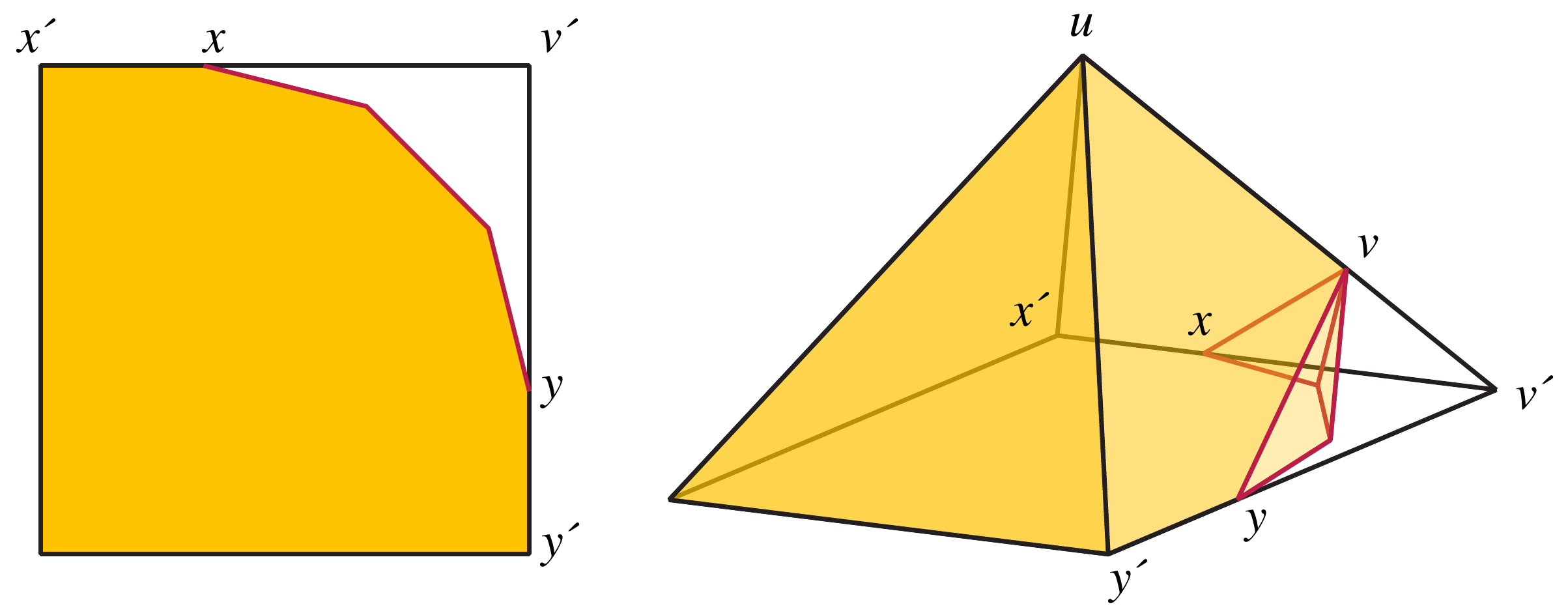}
\caption{Illustration for the proof of \autoref{prop:realize-halin}. Replacing a vertex $v'$ of the base polygon of a $3$-treetope by a convex chain, and adding a new vertex $v$ on the canopy edge $uv'$, produces a $3$-treetope whose canopy has one more internal vertex.}
\label{fig:halin-expansion}
\end{figure}

\begin{proof}
Let $G$ be a Halin graph, formed from a planar-embedded tree $T$ by adding a cycle connecting the leaves of $T$. We prove the result by induction on the number of internal vertices of~$T$.
As a base case, if $T$ has one internal vertex and $k\ge 3$ leaves, then $G$ can be realized as a pyramid over a regular $k$-gon. Otherwise, let $v$ be an interior vertex of $T$ with only one non-leaf neighbor~$u$. Let $\ell\ge 2$ be the number of leaf neighbors of $v$, and let $T'$ be the smaller tree formed from $T$ by replacing $v$ and its leaf neighbors by a single leaf vertex $v'$, adjacent to~$u$. By the induction hypothesis, the Halin graph formed from $T'$ has a realization as a $3$-treetope $P'$, in which $v'$ is a vertex of the base polygon $B'$. Let  $x'$ and $y'$ be the two neighbors of $v$ in this base polygon for $T'$.

Now, form a new base polygon $B$ from $B'$ by removing $v'$ and replacing it by a convex chain of $\ell$ vertices; call the first vertex in the chain $x$ and place it on edge $x'v'$, call the last vertex in the chain $y$ and place it on edge $y'v'$, and place the remaining vertices in convex position within triangle $xyv'$ (\autoref{fig:halin-expansion}, left). Place a new point for $v$ on edge $uv'$,
and let $P$ be the convex hull of the remaining vertices of $P$ together with the newly placed points (\autoref{fig:halin-expansion}, right).

Then, in the new polyhedron~$P$, edges $ux$, $x'x$, and $y'y$ lie on the same lines as the previous edges $uv'$, $x'v'$, and $y'v'$, so the change from $P'$ to $P$ does not change the link of any vertex that belongs to both polyhedra. However, in $P$, each vertex of the convex chain must have a neighbor outside $B$, for otherwise it would have a two-dimensional link. Therefore each of these vertices is connected to $v$ and to its two neighbors in $B$, but to no other vertices.
Thus, we have formed a $3$-treetope whose canopy now includes $v$ with the correct number of leaf neighbors, realizing~$G$ as required.
\end{proof}

A proof along the same lines was used by Aichholzer et al.~\cite{AicCheDev-CCCG-12} to prove that every Halin graph has a convex polyhedral realization in which the base face is horizontal and all other faces have equal slopes, or equivalently that every tree can be realized as the straight skeleton of a convex polygon.

We will realize our $4$-treetopes in the same way, by an inductive process in which we add one canopy vertex in each step. As in the above proof, the geometric placement of this canopy vertex will not be difficult: it can go anywhere along the parent edge of the leaf it replaces, and will automatically have edges connecting it to all the other vertices added in the same step. And, as in the above proof, all continuing vertices will have unchanged links, preventing them from having unwanted edges to newly added vertices. The part of the proof that is tricker is the replacement of $v'$ by a convex chain. In a two-dimensional base polygon, any placement of the correct number of vertices in convex position within the triangle $xyv'$ will work, because these all produce convex polygons with the correct number of sides. However, in the corresponding step for $4$-treetopes, we will need to replace a single vertex $v'$ of the base polyhedron by a set of vertices whose (three-dimensional) convex hull has a predetermined combinatorial structure. So placing the new vertices into convex position is not enough; they also need to be in positions with respect to each other that produce the correct three-dimensional convex hull.

\subsection{Face and cone shape realizability}

To achieve the desired placement of new vertices in each step of the inductive proof, we will use projective duality together with a known method for realizing convex polyhedra with specified face shapes.

\begin{lemma}[Barnette and Gr\"unbaum~\cite{BarGru-PJM-70}]
\label{lem:specified-2-face}
Let $G$ be a $3$-vertex-connected planar graph, $f$ be a $2$-face of its combinatorial embedding, and $B$ be a realization of $f$ as a convex polygon in the $xy$-plane of three-dimensional Euclidean space. Then there exists a three-dimensional polyhedron $P$ whose graph is isomorphic to~$G$, with $B$ as the face corresponding to $f$.
\end{lemma}

By applying a projective transformation that fixes the plane of $B$ we can additionally ensure that, for a given viewpoint $\alpha$ on the opposite side of that plane from $P$, $f$ is the only point of $P$ visible from $\alpha$. For our purposes we need a projectively dual version of this result, for which we need some more definitions.

% Way out of place in source again
\begin{figure*}[t]
\centering\includegraphics[scale=0.45]{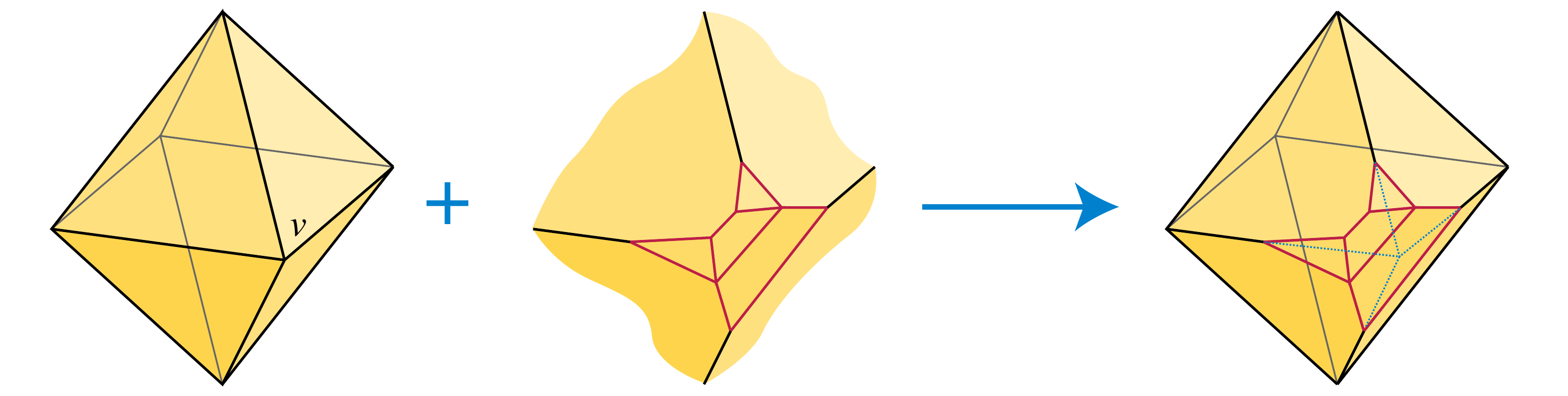}
\caption{Illustration for the proof of \autoref{thm:characterize}: replacing the vertex $v$ of the $3$-polyhedron $B'$ by a cone polyhedron whose cone matches the link of~$v$.}
\label{fig:cone-replace}
\end{figure*}

\begin{definition}
Suppose that finitely many halfspaces all have boundaries that pass through a common point~$p$, and that $p$ is the only point in the intersection of their boundaries. In such a case we call the intersection of the halfspaces a \emph{convex polyhedral cone}, and we call $p$ the \emph{apex} of the cone.
If $Q$ is the intersection of a convex polyhedral cone $C$ with finitely many additional halfspaces (none of which contain the apex~$p$), and every infinite face of $Q$ is a subset of an infinite face of the cone, we call $Q$ a \emph{cone polyhedron}, and we call $C$ the \emph{cone} of~$Q$. The faces of a cone polyhedron may be defined in the same way as for convex polyhedra. Equivalently, a cone polyhedron is an intersection of finitely many halfspaces with the property that the hyperplanes containing unbounded faces of the intersection intersect in a single point, the apex.

The \emph{graph} of a cone polyhedron $Q$ is an undirected graph with a vertex for each $0$-face of $Q$ together with one additional vertex, the \emph{cone vertex} of $Q$. It has an edge for each $1$-face of $Q$, of two types: a $1$-face that is a finite line segment connects two $0$-faces, and a $1$-face that is an infinite ray connects a $0$-face with the cone vertex.
\end{definition}

\begin{lemma}
\label{lem:cone-realization}
Let $C$ be a three-dimensional convex polyhedral cone with $k$ sides ($2$-faces), and let $G$ be a polyhedral graph with a designated vertex $v$ of degree $k$, and with a fixed order-preserving correspondence between the edges incident to~$v$ and the rays of the cone. Then there exists a cone polyhedron $Q$ whose graph is isomorphic to~$G$, such that the isomorphism maps $v$ to the cone vertex and respects the correspondence between edges and rays, and such that each infinite ray of $Q$ is a subset of the corresponding ray of $C$.
\end{lemma}

\begin{proof}
Let $\omega$ be a point interior to $C$ and let $\tau$ be a projective duality transformation that maps $\omega$ to the plane at infinity. Then $\tau$ maps the apex of $C$ to a non-infinite plane $\pi$, and it maps the planes through the sides of $C$ to points in convex position in $\pi$, forming the vertices of a convex polygon $B$. Additionally, $\tau$ maps the plane at infinity into a non-infinite point $\alpha$ that does not belong to plane $\pi$.

Apply \autoref{lem:specified-2-face} to realize the dual graph of $G$ as a polyhedron $P$ in which the face dual to $v$ is realized as polygon $B$, and additionally (by performing a projective transformation of $P$) arrange the realization in such a way that $P$ is on the other side of $\pi$ from $\alpha$ and $B$ is the only face of $P$ visible from $\alpha$. Then $\tau^{-1}(P)$ (a cell in the projective arrangement of hyperplanes dual to the vertices of $P$) has one vertex $v$ separated from all the others by the plane at infinity. That is, when viewed as a subset of Euclidean space, this cell in the arrangement has two connected components, one of which is a cone polyhedron and the other of which is a polyhedral cone (either $C$ or its reflection through the apex). The cone polyhedron component, reflected if necessary to lie within $C$, gives the desired realization $Q$.
\end{proof}

\subsection{Characterizing $4$-treetopes}

\begin{theorem}
\label{thm:characterize}
A graph $G$ is the graph of a $4$-treetope $P$ with base $B$ if and only if $G$ is the cluster graph of a well-connected clustering  $(C,F)$ of a polyhedral graph $F$, with $F$ forming the graph of $B$.
\end{theorem}

\begin{proof}
One direction, the claim that every $4$-treetope graph is a cluster graph, is \autoref{obs:treetope-cluster}. In the other direction, let $G$ be the cluster graph of a well-connected clustering $(C,F)$; we will prove by induction on the number of clusters that $G$ can be realized as a $4$-treetope. As a base case, if there is only one cluster (the set of all vertices of the base graph) then we may realize the base graph $F$ as a $3$-polytope $B$ by Steinitz's theorem, and then realize $G$ itself as the pyramid over~$B$.

Otherwise, by \autoref{thm:expansion}, let $(C,F)$ be obtained from a smaller well-clustered graph $(C',F')$ by an expansion operation. This operation replaces a vertex $v$ of $F'$ by a new cluster, whose cluster vertex may be called $c$. Let $H$ be the polyhedral graph used to form the expansion, and let $v'$ be the vertex that is removed from $H$ as part of the expansion (with $v$ and $v'$ having equal degrees). By induction, the cluster graph $G'$ of $(C',F')$ can be represented as a $4$-treetope $P'$ with base $B'$, with $F'$ isomorphic to the graph of $B'$.

Let $C$ be a polyhedral cone in the three-dimensional affine hull of $B'$, formed by the intersection of the $2$-faces of $B'$ that are incident to $v$.
By \autoref{lem:cone-realization} we may find a realization of $H$ as a cone polyhedron $K$, with $v'$ as the cone vertex and with $C$ as the cone of $K$, respecting the correspondence between neighbors of $v$ and neighbors of $v'$.
Scale this cone polyhedron to be small enough so that all of its infinite rays have starting points that lie within the edges of $B'$ incident to~$v$. Create a new base polyhedron $B$ by adding the vertices of the scaled cone polyhedron to $B'$ and removing~$v$ (\autoref{fig:cone-replace}).
Add another vertex representing $c$ anywhere on the edge from $v$ to its parent in $P'$, and compute $P$ as the convex hull of the set of vertices obtained in this way.

Then, in the new polytope~$P$, for each vertex $u$ that was a neighbor of~$v$ there exists a new vertex within line segment $uv$, so the change from $P'$ to $P$ does not change the link of any vertex that belongs to both $P$ and $P'$. However, in $P$, each vertex of $K$ must have a neighbor outside $B$, for otherwise it would have a two-dimensional link. Since all vertices with changed links belong to $B\cup\{c\}$, the only choice for a vertex outside $B$ to connect to is~$c$. Therefore each vertex of~$K$ has an edge to $c$, but to no other vertices outside~$B$.
Thus, we have formed a $4$-treetope whose canopy now includes~$c$, and where the base vertices reached from $c$ have the correct topology, realizing~$G$ as required.
\end{proof}

\section{Recognition}
\label{sec:recognition}

Our recognition algorithm for the graphs of $4$-treetopes is based on the idea of repeatedly finding and contracting a cluster in the clustering corresponding to the treetope. To this end, we seek the vertices that represent contractible clusters.

\subsection{Extremal vertices and extremal clusters}

\begin{definition}
Let $G$ be the graph of a $4$-treetope $P$ with base $B$. Then a vertex $v$ of~$G$ is \emph{extremal} if $v$ is disjoint from $B$ and has exactly one neighbor in $G$ that is also disjoint from $B$. An \emph{extremal cluster} of $G$ is the set of vertices consisting of $v$ and its neighbors in~$B$.
\end{definition}

Because the vertices and edges of a treetope that are disjoint from the base form a tree, whose leaves are the extremal vertices, we have:

\begin{observation}
\label{obs:pyramid-is-base-case}
Every $4$-treetope that is not a pyramid contains at least two extremal vertices.
\end{observation}

\begin{observation}
Let $G$ be the graph of a $4$-treetope $P$ with base $B$, and let $v$ be an extremal vertex of~$G$. Then $G$ is the cluster graph of a well-connected clustering $(C,H)$ where $H$ is the graph of~$B$, in which the neighbors in $B$ of $v$ form a minimal cluster in $H$.
\end{observation}

\begin{proof}
Choose a vertex $w$ of $B$ that is not a neighbor of $v$, and for each split of $P$ form a cluster in $B$ consisting of the base vertices in the branch of the split that does not contain~$w$.
\end{proof}

\begin{observation}
Let $G$ be the graph of a $4$-treetope $P$ with base $B$,  let $v$ be an extremal vertex of~$G$,
and let $G$ be the cluster graph of a well-connected clustering $(C,H)$ in which the extremal cluster of $v$ is one of the minimal clusters. Then the operation in $G$ of contracting the cluster of~$v$ into a single supervertex produces the cluster graph of the clustering formed by contracting the extremal cluster of~$v$.
\end{observation}

\subsection{Candidate vertices}

Intuitively, the overall outline of our algorithm will be to repeatedly identify and contract extremal clusters until reaching the graph of a pyramid, which is easily recognized. We would like to do this by using the properties of cluster graphs to identify their extremal vertices. However, these vertices cannot be uniquely identified, as the example of a tetrahedral prism demonstrates. This $4$-polytope has two tetrahedral facets and four triangular-prism facets; it can form a treetope in four different ways, with any one of the triangular-prism facets as its base and with the two remaining vertices that are outside this facet as its extremal vertices. Thus, in this polytope, every vertex is extremal, but not all choices of extremal vertices are compatible with each other. In other, larger $4$-treetopes, there can also exist vertices that are necessarily part of the base of the treetope, but whose local neighborhoods look like the neighborhoods of extremal vertices. Therefore, we define a broader class of vertices, the \emph{candidate vertices}, that include the extremal vertices and possibly some other non-extremal vertices.

\begin{definition}
Let $G$ be an arbitrary graph. We define a \emph{candidate vertex} to be a vertex $v$ of $G$ with the following properties:
\begin{itemize}
\item $v$ has at least four neighbors.
\item The graph induced in $G$ by the neighbors of $v$ is planar, and has exactly two connected components, one of which is an isolated vertex.
\item If $v$ is deleted from $G$, the nontrivial component of the neighbors of $v$ induces an externally $3$-vertex-connected subgraph of the remaining graph.
\item The set of edges connecting the vertices in the nontrivial component of the neighbors of $v$ to vertices (other than~$v$) outside this component forms a matching in $G$, with no two of these edges sharing an endpoint.
\end{itemize}
\end{definition}

\begin{observation}
The conditions for being a candidate vertex are checkable in polynomial time and are satisfied by every extremal vertex.
\end{observation}

Despite candidate vertices not necessarily being extremal vertices, they can be used to identify extremal clusters:

\begin{definition}
Let $G$ be the graph of a $4$-treetope $P$ with base $B$, and
let $v$ be a candidate vertex in~$G$. Then we define the cluster of~$v$ to be the set of vertices consisting of $v$ and the nontrivial connected component in the neighbors of~$v$.
\end{definition}

\begin{lemma}
\label{lem:candidates-are-good}
Let $G$ be the graph of a $4$-treetope $P$ with base $B$,
let $v$ be a candidate vertex, and let $Q$ be the cluster of~$v$. Then $Q$ is an extremal cluster for $G$ and~$B$.
\end{lemma}

\begin{proof}
We first observe that $v$ cannot be a non-base vertex that is not extremal, for every such vertex has a neighborhood that induces a graph with at least two isolated vertices (the canopy neighbors of~$v$).
And if $v$ is extremal, the result is true by definition. So the remaining case is that $v$ is a candidate vertex but that it belongs to $B$. In this case, let $u$ be the parent of~$v$, and let $w$ be the isolated vertex in the neighborhood of~$v$.

Then we have the following facts about $u$ and its cluster:
\begin{itemize}
\item \emph{The cluster of $u$ contains $v$.} This follows because otherwise $u$ would not be the parent of $v$.
\item \emph{The cluster of $u$ contains at least one neighbor $x$ of $v$ in $Q$.} For otherwise the only possible neighbor of $v$ in the cluster of $u$ would be its isolated neighbor~$w$, and the cluster could not be externally $3$-connected.
\item \emph{The cluster of $u$ contains every neighbor of~$v$ in $Q$.} For, if not, by the connectivity of the neighborhood, there would exist some two adjacent vertices $y$ and $z$ in $Q$ such that the cluster of $u$ contained $y$ but not $z$. But then there would be two edges from $z$ to the cluster of $u$ (one to $y$ and one to $v$), violating the requirement that no two edges into the cluster can share an endpoint.
\item \emph{The neighborhood of $u$ contains exactly one cluster vertex.} There must be at least one cluster vertex neighbor, for otherwise we would have a trivial clustering, and the neighbors of $v$ would form a single component connected through~$u$. And $u$ cannot have two cluster vertex neighbors, for that would violate the condition that the edges connecting neighbors of $v$ to the rest of the graph form a matching.
\item \emph{The neighborhood of $u$ does not contain any base vertex $z$ outside~$Q$.} For, if it did, $u$ would belong to $Q$ but would have two neighbors outside $Q$ (the vertex~$z$ and one cluster vertex neighbor), violating the requirement on the candidate vertex~$v$ that the edges from the cluster of $v$ to the rest of the graph form a matching.
\end{itemize}
We conclude from this chain of reasoning that $Q$ is  the cluster of~$u$ and that it is an extremal cluster.
\end{proof}

\begin{lemma}
\label{lem:contraction-reversal}
Let $G$ be an arbitrary graph, let $v$ be a candidate vertex, and let $Q$ be the cluster of $v$.
Let $G'$ be the graph formed by contracting $Q$ to a single supervertex $v'$, and suppose that $G'$ is the graph of a $4$-treetope in which $v'$ belongs to the base. Then $G$ is also the graph of a $4$-treetope in which $Q$ is an extremal cluster.
\end{lemma}

\begin{proof}
The reversal of the contraction operation can be interpreted as an expansion operation in a well-connected clustering whose cluster graph is~$G'$, taking it to a well-connected clustering whose cluster graph to~$G$. The result follows by \autoref{thm:characterize}.
\end{proof}

\subsection{The algorithm}

Based on the analysis of the previous sections, we can test whether a given graph $G$ is the graph of a $4$-treetope as follows:
\begin{itemize}
\item Initialize a set $K$ of known base vertices of $G$ to be the empty set
\item While $G$ contains a candidate vertex $v$ that does not belong to~$K$:
\begin{itemize}
\item Contract the cluster of $v$ into a single supervertex $v'$.
\item Add $v'$ to $K$.
\end{itemize}
\item If the remaining graph contains a universal vertex $u$ that does not belong to $K$, and the vertices other than $u$ induce a polyhedral graph, return yes. Otherwise, return no.
\end{itemize}

\begin{theorem}
The algorithm described above correctly tests whether its input is the graph of a $4$-treetope, in polynomial time.
\end{theorem}

\begin{proof}
Each step involves testing graph properties such as planarity that are already known to be polynomial, and each iteration of the loop reduces the size of the graph by at least one vertex, so the polynomial time bound for the algorithm is clear.

If $G$ is the graph of a $4$-treetope $P$ with base $B$, then by \autoref{lem:candidates-are-good} each iteration will correctly perform a contraction of an extremal cluster in $G$, and will correctly mark the resulting supervertex as part of the base of the contracted graph. Therefore, in this case, the algorithm will eventually reach a $4$-treetope that has no extremal vertex, which by \autoref{obs:pyramid-is-base-case} must be a $4$-pyramid. In such a graph, there does exist a universal vertex whose neighborhood is polyhedral, and the algorithm will correctly answer yes.

Conversely, suppose that the algorithm does answer yes. Then it will have found a sequence of contractions that reduce the given graph to the graph of a $4$-pyramid, whose apex does not belong to the set~$K$.
Then by \autoref{lem:contraction-reversal} each contraction made by the algorithm can be reversed to produce a $4$-treetope whose canopy is disjoint from~$K$. Therefore, the algorithm's "yes" answer is correct.
\end{proof}

We remark that recognizing the graphs of pyramids over arbitrary $4$-polytopes, and therefore also recognizing the graphs of $5$-treetopes, is as difficult as recognizing the graphs of arbitrary $4$-polytopes, which we expect to be complete for the existential theory of the reals.

\section{Sparsity of $4$-treetope graphs}
\label{sec:properties}

As unions of planar graphs and trees, the graphs of $4$-treetopes are necessarily sparse graphs.
But although the graphs of $3$-treetopes (the Halin graphs) have bounded treewidth, the same is not true for $4$-treetopes, because they include the graphs formed by adding an apex to arbitrary planar graphs. More strongly, as we show below, the $4$-treetopes are not contained in any nontrivial minor-closed graph family. Nevertheless, they obey stronger forms of sparsity than merely having a low ratio of edges to vertices. In particular, they have bounded expansion in the sense described by Ne{\v{s}}et{\v{r}}il and Ossona de Mendez~\cite{NesOss-Sparsity-12}.

\subsection{Knotted embeddings and arbitrary minors}

To prove that $4$-treetope graphs do not belong to any nontrivial minor-closed family, we study the knots, links, and graphs that can be embedded on the boundary of a $4$-polytope (topologically a $3$-sphere) using the vertices and edges of the polytope. For $4$-polytopes that are pyramids, the graph of the polytope cannot contain any nontrivial knots or links. For, in this case, every cycle either remains entirely on the base of the pyramid or it forms a path on the base together with two edges connecting the path endpoints to the apex. Thus, with respect to the $2$-sphere boundary of the base, it forms at most a $1$-bridge knot, which must therefore be the unknot. However, this does not extend to treetopes:

\begin{observation}
Let $K$ be an arbitrary knot or link in a topological $3$-sphere.
Then there exists a $4$-treetope $P$ whose graph contains a knot or link that is embedded into the boundary of $P$ in a way that is topologically equivalent to $K$.
\end{observation}

\begin{proof}
Draw a diagram of $K$ as a self-crossing curve in the plane, such that only two strands of the curve meet at each crossing point. Draw a circle surrounding each crossing point, small enough that it does not cross or contain any other such circle and intersecting the diagram of $K$ in exactly four crossing points. Add additional subdivision vertices along the strands of $K$ outside these circles, and edges between these vertices, as necessary so that the result becomes $3$-vertex-connected. Form a well-clustered graph by adding a cluster consisting of each original crossing point of $K$ and the four points where the circle surrounding it crosses $K$, and realize this clustering as a $4$-treetope. Then, at each crossing of $K$, one of the two strands of $K$ may be replaced by a two-edge path through the cluster vertex of the crossing, separating it from the other strand.
\end{proof}

A similar construction shows that, more generally, every embedding of a graph into three-dimensional space has a topologically equivalent embedding as a subgraph of a $4$-treetope within the boundary $3$-sphere of the treetope. In particular, this is true for embeddings of arbitrarily large complete graphs. Therefore, there are no forbidden minors for the graphs of $4$-treetopes. In this, again, the $4$-treetopes differ from the $4$-pyramids, for which the seven graphs of the Petersen family, and many others, are known forbidden minors~\cite{Pie-Chico-14}.

\subsection{Separators and bounded expansion}

A \emph{separator} of an $n$-vertex graph is a subset of vertices the removal of which partitions the remaining subgraph into connected components whose number of vertices is at most a constant fraction of~$n$. As is well known, planar graphs have separators of size $O(\sqrt n)$; this property forms the basis for many efficient algorithms for these graphs. We will use a stronger form of this planar separator theorem:

\begin{lemma}[Miller~\cite{Mil-JCSS-86,AloSeyTho-SJDM-94}]
\label{lem:simple-cycle-separator}
Every maximal planar graph has a separator of size $O(\sqrt n)$ that forms a simple cycle,
such that at most $2n/3$ vertices are inside the cycle and at most $2n/3$ vertices are outside the cycle.
\end{lemma}

As we will show, the graphs of $4$-treetopes also obey a similar separator theorem. To prove this, we define a superclass of the clustered planar drawings used to define $4$-treetopes.

\begin{definition}
We define a \emph{sparse clustering} to be a clustered planar drawing with the property that for each two disjoint vertex sets $X$ and $Y$ that are clusters, complements of clusters, or singleton vertex sets, and whose union is not the entire vertex set, at most one edge of $G$ has one endpoint in $X$ and one endpoint in $Y$. We define a \emph{sparse cluster graph} to be the cluster graph of a sparse clustering.
\end{definition} 

That is, we keep one of the main requirements of a well-connected clustering, but we forgo the requirements of minimum degree four per cluster vertex and external 3-vertex-connectivity of each cluster.

\begin{theorem}
\label{thm:separator}
Every $n$-vertex subgraph of a sparse cluster graph has a separator of size $O(\sqrt n)$. In particular this is true for the graphs of $4$-treetopes.
\end{theorem}

\begin{proof}
Let $G$ be a subgraph of a cluster graph of a sparse clustering.
We can assume that the whole cluster graph of the same clustering does not include any additional vertices, for if we had a clustering with additional vertices in the underlying graph or additional curves defining more cluster vertices than the ones in $G$, we could delete those vertices or curves from the clustering and obtain another clustering of which $G$ is also a subgraph. And since additional edges only make it more difficult to obtain a small separator, we may assume without loss of generality that $G$ is the whole cluster graph rather than a proper subgraph.

We define a planar graph $H$ from the clustering by the following steps:
\begin{itemize}
\item Augment the underlying planar graph by a new vertex at each crossing point of an edge and a cluster boundary (subdividing the edge at that point)
\item Add a cycle of edges that follow the curve surrounding each cluster, connecting the crossing points on that cluster. If there are only one or two crossing points, add a constant number of additional vertices on the curve so that it can be completed to a cycle of edges without forming a multigraph.
\item Complete the resulting planar embedded graph to a maximal planar graph, preserving its embedding.
\end{itemize}
The number of crossing points added on one of the curves of the clustering is at most proportional to the number of children of the corresponding cluster in the cluster hierarchy, from which it follows that $H$ has $O(n)$ vertices. By repeatedly applying \autoref{lem:simple-cycle-separator} we can find a collection of $O(1)$ cycles that together partition $H$ into connected subgraphs of at most $n/3$ vertices, and that each have length $O(\sqrt n)$.

We construct a separator that includes each vertex of the underlying planar graph that belongs to one of these separating cycles. We also include in the separator a cluster vertex for each region of the clustered drawing that is crossed by one of the separating cycles. Thus, the number of cluster vertices included in the separator is proportional to the number of crossing vertices of $H$ included in the separating cycles in $H$.
This separator partitions $G$ into subgraphs that correspond to the connected components of $H$. 

Each vertex in one of the connected components of $H$ that remain after removing the cycles either directly corresponds to a vertex of $G$ (if it is a vertex of the underlying planar graph) or to two cluster vertices in $G$ (if it is a crossing vertex in $H$). Thus, each of the remaining connected subgraphs of $G$ after the separator vertices are removed has at most $2n/3$ vertices.
\end{proof}

\begin{definition}
A $t$-\emph{shallow minor} of a given graph $G$ is another graph obtained from $G$ by contracting a collection of vertex-disjoint connected subgraphs of radius at most $t$, and then performing an arbitrary sequence of edge and vertex deletions on the result. A family of graphs has \emph{bounded expansion} if there exists a function $f$ such that every $t$-shallow minor of a graph in the family has a ratio of edges to vertices that is at most $f(t)$~\cite{NesOss-Sparsity-12}. More strongly, a family of graphs has \emph{polynomial expansion} if it has bounded expansion with a function $f$ that is bounded by a polynomial in~$t$~\cite{DvoNor-SSS-15}.
\end{definition}

These properties are important in some algorithmic applications. In particular, subgraph isomorphism is fixed-parameter tractable (parameterized by subgraph size) for graph families of bounded expansion~\cite{NesOss-Sparsity-12}, and several graph optimization problems including maximum independent set and minimum dominating set have polynomial-time approximation schemes on graphs of polynomial expansion~\cite{HarQua-ESA-15}.

\begin{theorem}
Subgraphs of sparse cluster graphs, and in particular the graphs of treetopes, have polynomial expansion.
\end{theorem}

\begin{proof}
The subgraphs of sparse cluster graphs form a hereditary graph class: a subgraph of a subgraph of a sparse cluster graph is itself a subgraph of a sparse cluster graph. As Dvo{\v{r}}{\'a}k and Norin~\cite{DvoNor-SSS-15} show, every hereditary graph class obeying a separator theorem with separator size $O(n^{1-\epsilon})$ for constant $\epsilon>0$ has polynomial expansion. The result follows from \autoref{thm:separator}.
\end{proof}

\section{Conclusions}

We have defined an interesting class of polytopes, generalizing the Halin graphs to higher dimensions. We have characterized the graphs of four-dimensional polytopes in this class in terms of the cluster graphs of certain clustered planar graph drawings, and used this characterization to develop polynomial time recognition algorithms for these graphs. We have also begun a preliminary graph-theoretic investigation of the properties of these graphs, showing that (unlike the graphs of Halin graphs and three-dimensional polyhedra) they do not have bounded treewidth or forbidden minors, but they do have bounded expansion.

Although our algorithms can be used to recognize the graphs of $4$-treetopes quickly, they do not immediately lead to a polynomial-time algorithm for constructing a realization of these polytopes. The problem is in our dependence on the Barnette--Gr\"unbaum realization of a three-dimensional polytope with a pre-specified face shape. Even if this specified face has small integer coordinates, the use of induction by Barnette and Gr\"unbaum may cause their realization to have doubly exponential coordinates, requiring an exponential number of bits to represent precisely. This issue naturally raises several questions: Can $4$-treetopes always be realized with integer coordinates (as $3$-polytopes can and general $4$-polytopes cannot)? If so, can we represent those coordinates using a polynomial number of bits per coordinate? And if that is also true, can we construct a realization in polynomial time?

\bibliographystyle{abuser}
\bibliography{treetope}

\smallskip % hack to prevent flushend from screwing up alignment of last bib line

\end{document}